\pdfoutput=1
\documentclass[letterpaper,journal]{IEEEtran}
\usepackage{amsmath,amsfonts}
\usepackage{algorithm}
\usepackage{array}
\usepackage[caption=false,font=normalsize,labelfont=sf,textfont=sf]{subfig}

\usepackage{textcomp}
\usepackage{stfloats}
\usepackage{url}
\usepackage{verbatim}
\usepackage{graphicx}
\usepackage{cite}
\usepackage{amsmath, amssymb}
\usepackage{graphicx}
\usepackage{algorithm,float}
\usepackage{algpseudocode}
\usepackage{bm}
\usepackage{tikz}
\newtheorem{remark}{Remark}
\newtheorem{lemma}{Lemma} 
\newif\ifshowrevisions
\showrevisionstrue  

\ifshowrevisions

\else

\fi

\begin{document}

\title{Amplitude-Tunable Pinching Antenna Systems: Single-Mode Phase-Mismatch Radiation and Multiuser Beamforming}

\author{Askin Altinoklu, and Leila Musavian%
\thanks{Askin Altinoklu and Leila Musavian are with the School of Computer Science and Electronic Engineering, University of Essex,  United Kingdom (e-mail: \{askin.altinoklu, leila.musavian\}@essex.ac.uk). This work was supported by UK Research and Innovation under the UK Government’s Horizon Europe funding guarantee through the MSCA-DN SCION Project, Grant Agreement No. 101072375 (Grant No. EP/X027201/1).}}

\maketitle

\begin{abstract}
Pinching antenna systems (PASS) enable reconfigurable radiating elements and extended line-of-sight communication, mitigating path loss effects. However, existing designs lack fully controllable radiation weights, as they are governed by structural parameters rather than explicitly assigned variables. In this paper, we introduce a new degree of freedom (DoF) for PASS by enabling radiation weight control through phase-mismatch manipulation of guided waves under single-mode excitation within a coupled-mode framework. By tuning the propagation constants of pinching antennas, independent complex-weight control of individual elements is achieved, transforming PASS into a weight-adaptive analog beamforming architecture. Based on this principle, we present a physics-based hardware model that provides a unified framework for both amplitude-tunable pinching beamforming and conventional equal-power radiation models, ensuring compatibility with existing PASS implementations, such as movable setups. To evaluate the proposed model, we formulate a sum-rate maximization problem for hybrid precoding in multiuser downlink systems and solve it using an alternating optimization framework that combines weighted minimum mean square error-based digital precoding with genetic algorithm-based optimization of PASS configurations, including various scenarios such as weight tuning, antenna movability, and discrete activation. Numerical results demonstrate that the amplitude-tunable PASS architecture achieves consistent performance gains over conventional arrays and existing PASS schemes, with pronounced improvements in interference-limited regimes under practical constraints.

\end{abstract}


\section{Introduction}
\label{sec:Intro}

\IEEEPARstart{F}{uture} sixth-generation (6G) wireless systems are envisioned to support extremely high data rates beyond the capabilities of current antenna array technologies \cite{6g_kpi3}. Meeting these demands requires substantially larger antenna apertures and operation over wider bandwidths, motivating the evolution of multiple-input multiple-output/single-output (MIMO-MISO) arrays toward extremely large-scale arrays. Increasing the carrier frequency enables denser antenna packing and higher spatial resolution along with the advantage of availability of wider bandwidths, but also introduces severe propagation attenuation for wireless networks \cite{gMIMO3}. Conventional MIMO architectures can compensate for this loss by increasing the number of antenna elements, which entails a proportional growth in radio-frequency chains and power amplifiers. As the array size scales to very large apertures, hardware complexity, power consumption, and cost of hardware structure become dominant bottlenecks, ultimately limiting for the scalability of conventional array-based beamforming architectures \cite{gMIMO1}. To overcome these limitations, recent architectures such as holographic and tri-hybrid MIMO enable large-aperture beamforming with reduced hardware complexity by leveraging analog-domain signal processing to decouple spatial DoF from the number of RF chains \cite{hMIMO1, trihybrid_mimo}.

Despite their advantages, these state-of-the-art architectures rely on pre-configured radiation apertures, limiting their ability to adapt radiation locations to the propagation environment. To address this limitation, a promising concept for generating reconfigurable wireless communication environments has recently emerged with the technology proposed and demonstrated by NTT DOCOMO \cite{suzuki2022pinching}, namely Pinching-Antenna Systems (PASS). This system paves the way for a fundamentally different paradigm in which radiation locations are externally reconfigurable rather than confined to fixed apertures. 

PASS consists of a dielectric waveguide that guides high-frequency EM signals, where energy is radiated through small particles, referred to as pinching antennas, placed in close proximity to the waveguide \cite{Intro2_pinching}. Its operation follows coupled-mode theory: a secondary waveguide couples to the evanescent field of the main guide and extracts power. Under the single-mode operation, the power transfer between the guides is characterized by three parameters: the coupling coefficient, the transfer length, and the propagation-constant mismatch between their propagation constants \cite{direc_coupler_3}. The transfer length corresponds to the antenna length that enables complete power transfer under phase-matched propagation, whereas the coupling coefficient mainly depends on the separation between the guides and their EM properties \cite{direc_coupler_1}. Accordingly, the radiated power of individual pinching antenna is determined by the coupling condition established by the antenna geometry and waveguide interaction.

Extending this mechanism to multiple pinching points enables radiation and power extraction at multiple locations. Hence, PASS has been presented as a flexible antenna system that offers advantages over existing flexible structures such as fluid and movable antennas in mitigating large-scale path loss with low cost and reduced structural complexity \cite{Intro1_pinching}. Since the baseband signal is guided through a dielectric waveguide, PASS inherently supports analog-domain signal control with guided-wave apertures. Furthermore, its flexibility in controlling externally attached pinching antennas enables a beamforming mechanism often referred to as pinching beamforming, allowing reconfiguration through adding, removing, or relocating radiating elements beyond the capabilities of conventional movable and fixed array architectures \cite{Intro3_pinching}. Owing to the low-loss propagation ability of the state-of-the art dielectric waveguides, PASS has been demonstrated for its ability of extending effective line-of-sight (LoS) communication over tens of meters, presenting the opportunity for alleviating path-loss attenuation and severe fading at high frequencies \cite{attenuation_of_pass}.

With aforementioned advantages of PASS, several transmission strategies have been developed to employ it in wireless communication networks. In particular, when multiple waveguides are combined with digital precoding, the joint design of baseband precoders and pinching beamforming enables multiuser transmission, commonly referred to as waveguide multiplexing \cite{wave_multiplex_7}. In this architecture, each dielectric waveguide is connected to a dedicated RF chain and controlled by a baseband processor. Pinching antennas are typically assumed to radiate equal power, while beamforming is realized through antenna activation or movement \cite{wave_multiplex_2}. Under this framework, joint optimization of digital precoding and pinching beamforming has been demonstrated to outperform conventional fixed-aperture hybrid architectures in multiuser downlink scenarios for objectives such as transmit power minimization, sum-rate maximization, and minimum-rate maximization \cite{wave_multiplex_1,wave_multiplex_2,wave_multiplex_3,wave_multiplex_5,wave_multiplex_6,wave_multiplex_7}. Recent studies further extend waveguide multiplexing to multi-mode PASS configurations, where multiple guided modes within a single waveguide carry different data streams and pinching antennas with distinct propagation constants selectively couple to specific modes for multiuser transmission \cite{pass_hardware_survey, multimode_pass}.

Despite promising communication-level results, most existing works treat the radiation of pinching antennas abstractly and do not explicitly utilize the coupled-mode physics governing power transfer. A limited number of studies account for the physical structure of PASS elements \cite{wave_multiplex_2,wave_multiplex_3,wave_multiplex_4_atten,multimode_pass}. In these works, equal-power radiation is achieved by pre-designing transfer lengths based on wave propagation \cite{wave_multiplex_2,wave_multiplex_4_atten}, while other approaches adjust radiation through transverse movement of pinching antennas that modifies the coupling coefficient \cite{wave_multiplex_3}. The same structural control has also been employed in multi-mode PASS to enable mode-selective coupling between antennas and guided modes for multiuser transmission \cite{multimode_pass}. More flexible schemes introduce power allocation factors using movable antennas \cite{wave_multiplex_5}; however, the radiation parameter is not linked to a physically tunable parameter governed by the coupled-mode model. Consequently, PASS does not provide independently controllable complex radiation weights, limiting the effectiveness of pinching beamforming, particularly in waveguide multiplexing. This limitation becomes clear when compared with conventional analog beamforming, where adjustable complex weights across antenna elements are realized using external phase shifters or reconfigurable components such as varactor-based metasurface arrays \cite{hMIMO1,trihybrid_mimo}. Such architectures achieve an effective precoder dimension larger than the number of RF chains by directly assigning element weights. In contrast, pinching beamforming lacks an equivalent controllable weighting stage because radiated amplitudes are determined by coupling conditions rather than assigned variables. As a result, existing approaches mainly optimize antenna placement or activation rather than performing true amplitude-tunable analog beamforming.

Enabling analog beamforming with full spatial DoF via pinching antennas requires a physically tunable parameter that can be controlled rapidly and precisely. Current PASS designs regulate radiation by adjusting the coupling coefficient or the effective interaction length under phase-matched propagation; however, both mechanisms offer limited controllability for realizing amplitude-tunable analog beamforming. Motivated by these limitations, we propose a new PASS architecture that enables radiation weight control through phase-mismatch manipulation of guided waves within a coupled-mode framework under single-mode operation. In guided-wave physics, phase-mismatch between coupled waves is a well-known mechanism for regulating power transfer and forms the basis of electrically controlled optical switches and modulators \cite{direc_coupler_2,direc_coupler_3}. Recent advances in tunable-permittivity materials demonstrate electrically adjustable propagation characteristics at radio frequencies. In particular, liquid-crystal-loaded dielectric waveguides allow voltage-controlled rotation of anisotropic molecules, enabling continuous tuning of the effective permittivity and the propagation constant of the guided mode. These structures have been widely used to realize electrically tunable phase shifters and beam steering devices \cite{lqc_1, lqc_2} as they provide fast and precise tunability in such phased array beamforming structures. Despite this feasibility, the use of electrically controlled phase-mismatch for regulating radiation in PASS communication systems has not yet been investigated. To date, the radiation coefficient of a pinching antenna remains either a mechanically controlled parameter, adjusted through antenna movement affecting the coupling coefficient \cite{wave_multiplex_3}, or a structural parameter determined by the physical length of the antenna \cite{wave_multiplex_2,wave_multiplex_4_atten}.

The proposed mechanism enables amplitude-tunable pinching beamforming by controlling the propagation constants of pinching antennas, thereby introducing controllable phase-mismatch between the guided modes of the source waveguide and the pinching antennas. This phase-mismatch controllability enables tunable complex weights of the radiation coefficients. Building on the aforementioned advances in material tunability, the mechanism can be practically realized through electrically controllable material properties. Consequently, PASS evolves from a purely geometry-controlled radiation mechanism into a weight-adaptive analog beamforming architecture. This transformation acts as a performance multiplier, facilitating the integration of PASS into state-of-the-art hybrid architectures such as tri-hybrid systems \cite{pass_hardware_survey}. By utilizing controllable weighting, PASS enhances its beamforming DoF while retaining its inherent advantages, including large-area deployment capability, “last-meter” proximity to users, and the movable reconfigurability of its radiating elements. This approach enables independent complex-weight control of individual pinching antennas, significantly extending the practicality of existing PASS implementations, including those based on movable antenna concepts. In particular, conventional equal-power radiation, commonly assumed in the literature, emerges as a special case of the proposed framework. More importantly, the method introduces a new paradigm of amplitude-tunable pinching beamforming, unlocking additional flexibility for spatial domain multiple access and performance optimization.

Building upon this new concept, the main contributions of this paper are summarized as follows:
\begin{itemize} 
 \item We propose a physics-based amplitude-tunable hardware model for PASS, enabled by controlling the propagation constants of pinching antennas. Using coupled-mode theory, we establish the analytical relationship between the amplitude–phase response of each element and variations in propagation constants of pinching antennas, thereby linking their tunable material properties such as electrical permittivity to radiation controls. The proposed model enables electrically manageable and adjustable radiation of pinching antennas along the main waveguide.
 \item We apply the derived hardware model to the signal modeling of multi-element PASS configurations. The introduction of phase-mismatch via tunable propagation constants induces simultaneous phase and amplitude variations in the radiation of individual elements. This enables a generalized multi-element signal model that captures amplitude-tunable pinching beamforming. As a special case, we obtain the conventional equal-power radiation model by enforcing analytically calculated phase-mismatch conditions on the propagation constants, thereby covering widely used movable and discrete activation PASS schemes within this unified framework.
\item To validate the proposed model, we investigate its performance in a multi-user downlink network. We formulate an optimization problem for digital precoding and pinching beamforming aimed at sum-rate maximization under transmit power constraints. To address the resulting non-convex problem, we adopt a hybrid solution strategy: for a given PASS array configuration, the digital precoders are optimized using the well-established weighted minimum mean square error (WMMSE) algorithm, while the PASS configuration is subsequently refined via a search-based heuristic approach employing genetic algorithms guided by the WMMSE-optimized precoders. 
\item Owing to the unified hardware framework of the proposed pinching antenna model, we extend the optimization algorithm to include conventional PASS schemes, such as movable and discrete activation configurations based on equal-power radiation. This enables a systematic comparison with existing PASS strategies. Numerical results show that the proposed radiation control mechanism enhances hybrid precoding capability, yielding consistent performance gains over existing PASS schemes and conventional fixed arrays, with more pronounced improvements in high inter-user interference regimes. This is due to the additional DoF provided by tunable complex weights. We further examine practical limitations, including material tunable ranges and weight quantization, demonstrating the feasibility of the proposed amplitude-tunable PASS architecture.

\end{itemize}

\IEEEpubidadjcol

\noindent\textit{Notation:} 
Matrices and vectors are denoted by boldface capital and lower-case letters, respectively. For vectors and matrices, $(\cdot)^T$ and $(\cdot)^H$ denote the transpose and Hermitian transpose, respectively. For a vector $\mathbf{w}$, $\|\mathbf{w}\|$ denotes the Euclidean norm. Scalars are denoted by lower-case letters, and $|\cdot|$ denotes the absolute value. The sinc function is defined as $\mathrm{sinc}(x) = \sin(\pi x)/(\pi x)$.

\section{PHASE-MISMATCH RADIATION CONTROL AND SIGNAL MODELING FOR PASS}

In this section, we derive a physics-based model for radiation control of a pinching antenna based on the phase mismatch between the eigenmode propagation constants of the source waveguide and the antenna. The power transfer is characterized using the transmission parameters of an equivalent two-port system as a function of this phase mismatch. By tuning the antenna propagation constant via its material properties, the phase mismatch becomes controllable, enabling direct control of the radiated power and its associated complex radiation weight. The formulation is then extended to multiple pinching antennas coupled to a single-mode waveguide, leading to the corresponding PASS signal model.

\subsection{Single-Antenna Radiation Model with Phase-Mismatch-Induced Amplitude Control}
For modeling, we adopt the directional-coupler representation \cite{pass_hardware_survey}, in which the pinching antenna is modeled as an open-ended secondary waveguide placed at a small fixed spacing from the main source waveguide. Under single-mode operation, the close proximity of the two waveguides enables overlap of their evanescent fields. The resulting field interaction can be described using coupled-mode theory under the weak-coupling assumption \cite{wave_multiplex_2,direc_coupler_1}. 

Assuming single-mode propagation along the $z$-direction in both waveguides, $\beta_1$ and $\beta_2$ denote the propagation constants of the main waveguide and the pinching antenna, respectively, and $A_1(z)$ and $A_2(z)$ represent the corresponding complex modal amplitudes coupled along the propagation direction. Then, the electric fields in the main waveguide and the pinching antenna can be expressed as
\begin{subequations}\label{eq:Eq1}
\begin{align}
\mathbf{E}_1(x,y,z) &= \mathbf{E}_1(x,y) A_1(z) e^{-j \beta_1 z}, \label{eq:Eq1a} \\
\mathbf{E}_2(x,y,z) &= \mathbf{E}_2(x,y) A_2(z) e^{-j \beta_2 z}. \label{eq:Eq1b}
\end{align}
\end{subequations}
Here, the propagation constants are related to the material properties as $\beta_1 = k_0 n_1$ and $\beta_2 = k_0 n_2$, where $n_1$ and $n_2$ are the refractive indices,  $k_0 = \frac{2\pi}{\lambda_0}$ denotes the free-space wavenumber with $\lambda_0$ denoting the free-space wavelength.
Then, based on coupled-mode theory, the equations governing change of complex amplitudes in $z$-direction can be written:

\begin{subequations}\label{eq:Eq3}
\begin{align}
\frac{d A_1(z)}{d z}
&= -j \, \kappa_{21} \, e^{j \Delta \beta z} \, A_2(z),
\label{eq:Eq3a} \\
\frac{d A_2(z)}{d z}
&= -j \, \kappa_{12} \, e^{-j \Delta \beta z} \, A_1(z),
\label{eq:Eq3b}
\end{align}
\end{subequations}
where $\Delta \beta = \beta_1 - \beta_2$ denotes the phase-mismatch per unit length between the propagation constants of the two guides. The coupling coefficients in \eqref{eq:Eq3}, denoted by $\kappa_{12}$ and $\kappa_{21}$, are determined by the transverse field distributions and physical parameters of the structure, such as the spacing between the waveguides \cite{wave_multiplex_3}. 
Then, by applying boundary conditions as  $A_1(0)=1$ and  $A_2(0)=0$, the solution for coupled equations can be derived as\cite{direc_coupler_1}:
\begin{subequations}\label{eq:Eq5}
\begin{align}
A_1(z) &=
e^{j\Delta \beta z /2}
\left(
\cos(\gamma z)
- j \frac{\Delta \beta}{2\gamma}
\sin(\gamma z)
\right), \label{eq:Eq5a} \\
A_2(z) &=
\frac{\kappa_{21}}{j\gamma}
\, e^{-j\Delta \beta z /2}
\sin(\gamma z), \label{eq:Eq5b}
\end{align}
\end{subequations}
where $\gamma = \sqrt{\kappa_{12}\kappa_{21} + \left(\frac{\Delta \beta}{2}\right)^2}$ denotes the eigenvalue of the coupled system.
Then, the total power of guided waves within the main waveguide $P_1(z)$ and the pinching antenna $ P_2(z)$ in relation to the initial source power $ P_1(0)$ can be written by:
\begin{subequations}\label{eq:Eq8}
\begin{align}
P_1(z) &= P_1(0)
\left[
\cos^2(\gamma z)
+ \left(\frac{\Delta \beta}{2\gamma}\right)^2
\sin^2(\gamma z)
\right], \label{eq:Eq8a} \\
P_2(z) &= P_1(0)
\frac{|\kappa_{21}|^2}{\gamma^2}
\sin^2(\gamma z). \label{eq:Eq8b}
\end{align}
\end{subequations}
The derived expressions indicate that energy is periodically exchanged between the main waveguide and the pinching antenna along the propagation direction. In the special case where the waveguide and the pinching antenna have identical propagation constants $\beta_1 = \beta_2$, with equal refractive indices $n_1 = n_2$, and assuming symmetric coupling $\kappa = \kappa_{12} = \kappa_{21}$, complete power transfer from the waveguide to the pinching antenna occurs at $z = L_0 = \frac{\pi}{2\kappa}$, referred to as the transfer length \cite{wave_multiplex_2}.  This condition corresponds to the phase-matched case, which constitutes the primary approach available to date in existing PASS implementations \cite{wave_multiplex_2,wave_multiplex_3,wave_multiplex_4_atten}. In contrast, we extend this framework to capture the effect of phase mismatch, thereby enabling controllable radiation power.

Consider a symmetric directional coupler model with $\kappa_{12} = \kappa_{21} = \kappa$ and a fixed pinching-antenna length $L_0 = \frac{\pi}{2\kappa}$. Transmission parameters of the source waveguide $T_{11}$  and pinching antenna $T_{21}$ can be derived as a function of phase-mismatch by inserting $z = L_o = \frac{\pi}{2\kappa}$ into \eqref{eq:Eq5} as:

\begin{subequations}\label{eq:Eq10}
\begin{align}
T_{11} &= 
e^{j \Delta \beta L_o / 2}
\left(
\cos\!\left( \frac{\pi}{2} \theta \right)
- j \frac{\Delta \beta L_o}{2}
\text{sinc}\!\left(\frac{\theta}{2}\right)
\right), \label{eq:Eq10a} \\
T_{21} &=
\frac{\pi}{2}
e^{-j \pi / 2}
e^{-j \Delta \beta L_o / 2}
\text{sinc}\!\left(\frac{\theta}{2}\right), \label{eq:Eq10b}
\end{align}
\end{subequations}
where $\theta = \sqrt{1+\left( \frac{\Delta \beta L_0}{\pi} \right)^2}$. Here, $T_{11}$ and $T_{21}$ denote the two-port transmission coefficients of the coupled waveguide--antenna structure under excitation from the main waveguide. The coefficient $T_{11}$ corresponds to the through field in the main waveguide, while $T_{21}$ measures the field coupled to the pinching antenna and radiated into free space. As observed from \eqref{eq:Eq10}, introducing a phase-mismatch induces both amplitude and phase variations in the guided field. 
Hence, this mechanism enables controllable radiation through the phase mismatch $\Delta \beta$, where the resulting normalized power transfer function, denoted by $\mathcal{T}$, from the main waveguide to the pinching antenna is given by
\begin{equation}
\mathcal{T}(\Delta \beta)
=
\frac{\pi^2}{4}\,
\mathrm{sinc}^2\!\left(\frac{\theta}{2}\right).
\label{eq:Eq12}
\end{equation}

\begin{figure}[t]
    \centering
    \includegraphics[width=0.7\columnwidth]{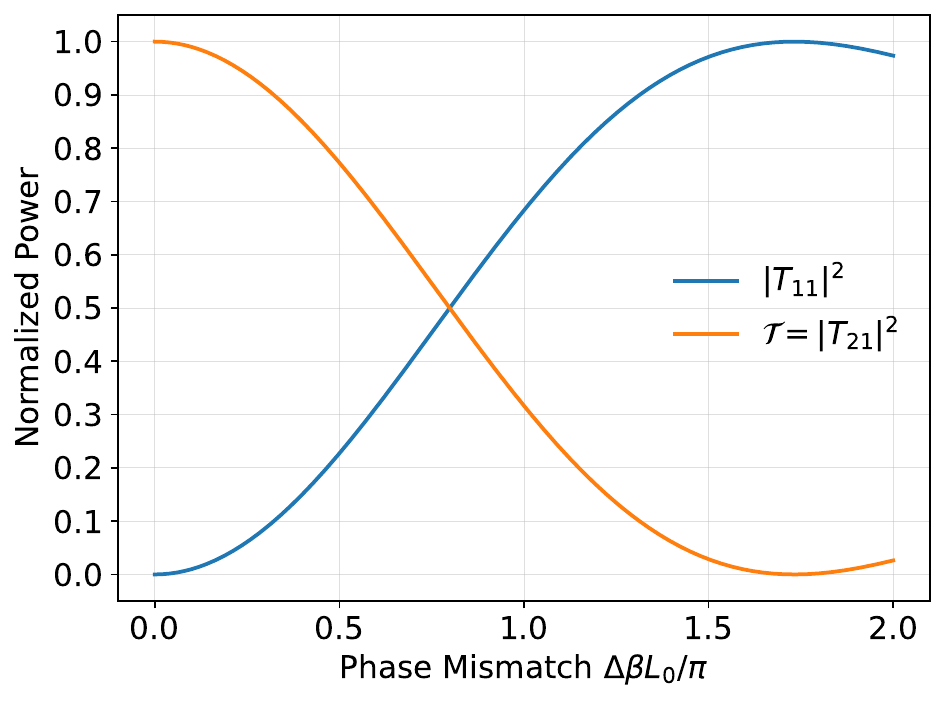}
    \caption{Power distributions in the pinching antenna and the waveguide  in response to phase-mismatch.}
    \label{fig:Fig2}
\end{figure}
This expression characterizes the radiated power through the pinching antenna, exhibiting a sinc-based variation with respect to $\Delta \beta$. This characteristic, along with the remaining power in the waveguide, given by $|T_{11}|^2 = 1 - \mathcal{T} = 1 - |T_{21}|^2$, is illustrated in Fig.~\ref{fig:Fig2}. It can be observed that for $0 < \Delta \beta L_0 < \pi \sqrt{3}$, the power transfer to the pinching antenna $\mathcal{T}$ varies continuously between 0 and 1, enabling controllable amplitude profile of the radiated power. In particular, $\Delta \beta L_0 = 0$ corresponds to the phase-matched condition, under which complete power transfer to the pinching antenna occurs. In contrast, $\Delta \beta L_0 = \pi \sqrt{3}$ leads to complete suppression of radiation, with the power remaining entirely in the  main waveguide.
This two-port representation establishes the fundamental radiation mechanism of a single pinching antenna under phase-mismatch control and provides, for the first time, an explicit closed-form relationship between phase-mismatch and the resulting complex radiation characteristics, enabling tunable radiation model.

\subsection{Multi-Antenna Signal Model with Phase-Mismatch-Induced Amplitude Control}

Building on this foundation, we extend the single-antenna formulation to multiple pinching antennas, where the cascaded interaction of transmission coefficients leads to a structured and physically consistent signal model for PASS. To this end, we consider the setup illustrated in Fig.~\ref{fig:pass_schematic}, consisting of a dielectric waveguide supporting single-mode propagation with propagation constant $\beta_\text{g} = k_0 n_{\text{g}}$ and attenuation constant $\alpha_\text{g}$, where $n_{\text{g}}$ is the effective refractive index of the guided mode. In this setup, $N_\text{a}$ pinching antennas are positioned along the waveguide at longitudinal locations $\{z_n\}_{n=1}^{N_\text{a}}$, where $z_n$ denotes the position of the $n$-th pinching antenna. Furthermore, all pinching antennas are assumed to have identical lengths $L_0$ and are placed at a fixed transverse separation from the source waveguide, resulting in identical coupling coefficients $\kappa_0$. We further assume that the coupling coefficient and antenna length satisfy $\kappa_0 L_0 = \frac{\pi}{2}$. Each pinching antenna supports a guided mode with propagation constant $\beta_\text{p} = k_0 n_{\text{p}}$, where $n_{\text{p}}$ is the effective refractive index of the pinching antenna. We assume that $n_{\text{p}}$ is electrically tunable through variations in the effective permittivity or permeability of the antenna material \cite{lqc_2}. Accordingly, each pinching antenna experiences a distinct phase-mismatch $\Delta \beta_n = \beta_{\text{g}} - \beta_{\text{p},n}$. By tuning $n_{\text{p}}$, these phase-mismatches can be adjusted from phase-matched (fully activated) to mismatched (deactivated) conditions, thereby enabling continuous control of the radiated power via \eqref{eq:Eq12}, as illustrated in Fig.~\ref{fig:pass_schematic}.
\begin{figure}[t]
    \centering
    \includegraphics[width=0.9\columnwidth]{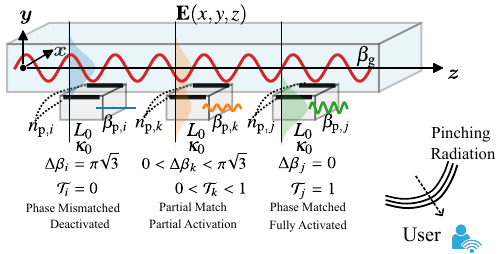}
    \caption{Illustration of PASS with amplitude-tunable hardware model.}
    \label{fig:pass_schematic}
\end{figure}

To derive the signal model for multiple elements, we first define the in-waveguide propagation of the unperturbed guided wave as
\begin{equation}
g_n = e^{-(\alpha_\text{g} + j\beta_\text{g}) z_n},
\label{eq:Eq13}
\end{equation}
where $z_n$ denotes the longitudinal position of the $n$-th pinching antenna from the feed of the waveguide.
Based on the transmission parameters defined in \eqref{eq:Eq10}, the cascaded transmission coefficient $a_n$ of the $n$-th antenna, which captures the effect of all preceding pinching antennas, can be expressed as
\begin{equation}
a_n
=
\left( \prod_{i=1}^{n-1} T_{11}^{(i)} \right) 
T_{21}^{(n)},
\quad n = 1,\ldots,N_\text{a} .
\label{eq:Eq14}
\end{equation}
Accordingly, the complex signal radiated by the $n$-th pinching antenna is given by $x_n = a_n g_n$, and the corresponding radiated power is $P_n = |a_n|^2 e^{-2\alpha_g z_n}$. Then, assuming a lossless transmission model for each pinching antenna such that 
$|T_{21}^{(n)}|^2 + |T_{11}^{(n)}|^2 = 1$, an adjustable power radiation model 
based on controllable phase-mismatch can be expressed as
\begin{equation}
\mathcal{T}(\Delta\beta_n)\, e^{-2\alpha_g z_n}
=
|T_{21}^{(n)}|^2 e^{-2\alpha_g z_n}
=
\frac{P_n^{\mathrm{desired}}}
{\displaystyle
\prod_{i=1}^{n-1}
\left(1-\mathcal{T}(\Delta\beta_i)\right)},
\label{eq:Eq17}
\end{equation}
where $P_n^{\mathrm{desired}}$ denotes the desired radiated power of the 
$n$-th pinching antenna. This controllable power radiation model can be used to derive different power radiation scheme of PASS, as formalized in the following lemmas.

\begin{lemma}[Amplitude-Tunable Pinching Antennas]
\label{AT-PASS}
Let the vector of tunable wavenumber phase-mismatches of all antennas be defined as
\(
\boldsymbol{\Delta\beta} =
\begin{bmatrix}
\Delta\beta_{1} & \Delta\beta_{2} & \cdots & \Delta\beta_{N_\text{a}}
\end{bmatrix}^{T}.
\)
Then, the amplitude and phase of the tunable complex radiation weights of the antennas can be expressed as
\begin{subequations}\label{eq:Eq18}
\begin{align}
|T_{21}^{(n)}|
&=
\frac{\pi}{2}\,\mathrm{sinc}\!\left(
\frac{1}{2}
\sqrt{
1 + \left( \frac{\Delta \beta_n L_0}{\pi} \right)^2
}
\right), \label{eq:Eq18a} \\
\angle T_{21}^{(n)}
&=
-\left(\frac{\pi}{2}+\frac{\Delta\beta_n L_0}{2}\right),
\quad n=1,\ldots,N_\text{a}. \label{eq:Eq18b}
\end{align}
\end{subequations}

The feasible set of phase-mismatch parameters for a fixed pinching-antenna length $L_0$ is defined as
\begin{equation}
\mathcal{B} =
\left\{ \Delta\beta_n \;\middle|\; 0 \le \Delta\beta_n L_0 \le \pi\sqrt{3} \right\},
\quad n=1,\ldots, N_\text{a},
\label{eq:Eq20}
\end{equation}
which enables full control over the dynamic range of antenna amplitudes according to \eqref{eq:Eq12}. This formulation leads to the amplitude-tunable pinching antennas proposed in this work.
\end{lemma}

\begin{remark}
For the phase-mismatch range $0 < \Delta\beta < \pi\sqrt{3}/L_0$, the power-transfer function $\mathcal{T}$ varies between $0$ and $1$, enabling continuous control of the radiated power. In addition, $\Delta\beta_n$ induces a phase shift in the radiated signal, as given in \eqref{eq:Eq18b}. Hence, its tunability enables joint amplitude–phase control, allowing analog beamforming with complex radiation weights and extending pinching-antenna operation beyond conventional phase-matched designs, with functionality comparable to tunable-metasurface-based hybrid architectures.
\end{remark}

Amplitude controllability of pinching antennas over $\Delta \beta$ presented in \eqref{eq:Eq18} can be extended to special case of Equal-Power Radiation, which is widely used with movable and discrete activation schemes. 
\begin{lemma}[Equal-Power Radiation via Phase-Mismatch]
\label{EPA-PASS}
Let $\mathcal{M}\subseteq\{1,\dots,N_\text{a}\}$ denote the index set of the $M$ active
pinching antennas. Define the corresponding binary antenna activation vector
$\mathbf{b}\in\{0,1\}^{N_\text{a}}$, where $b_n=1$ if the $n$-th antenna is active and
$b_n=0$ otherwise. Then, the phase-mismatch vector
$\boldsymbol{\Delta\beta}\in\mathbb{R}^{N_\text{a}}$ can be constructed as

\begin{equation}
\label{eq:21}
\boldsymbol{\Delta\beta} =
\begin{cases}
{\Delta\beta_n^{\mathrm{equal}}}   & \text{if }  b_n=1, n\in\mathcal{M} \\
0 & \text{if }  b_n=0, n\notin\mathcal{M}.
\end{cases}
\end{equation}
In \eqref{eq:21}, the phase-mismatch values for the active antennas follow the equal-power radiation model and are computed as
\begin{equation}
\Delta\beta_m^{\mathrm{equal}}
=
\frac{\pi}{L_0}
\sqrt{
\left(
2
\,\mathrm{sinc}^{-1}
\!\left(
\frac{2 \delta_m}{\pi}
\right)
\right)^2
-1
},
 \forall m \in \mathcal{M},
\label{eq:Eq22}
\end{equation}
where
\[
\delta_m = \sqrt{\frac{P_{\mathrm{equal}}}{e^{-\alpha_g z_m}-(m-1)P_{\mathrm{equal}}}},
\qquad
P_{\mathrm{equal}}=\frac{1}{M}.
\]

\textit{Proof:} See Appendix.
\begin{remark}
The equal-power radiation model represents the commonly adopted scenario for pinching antennas in the literature, where physics-based formulations are often overlooked in optimization. Existing coupled-mode approaches achieve equal-power radiation under phase-matched conditions by varying antenna length \cite{wave_multiplex_2,wave_multiplex_4_atten} or by moving antennas away from the waveguide to adjust the coupling \cite{wave_multiplex_3, multimode_pass}. However, such mechanisms limit practical applicability, as antenna activation or deactivation requires structural modifications or precise control of the antenna–waveguide separation, which cannot be achieved with fast and continuous response. In contrast, the equal-power model in Lemma~2, derived from phase-mismatch radiation, enables dynamic and practical control through electrically tunable material properties. It therefore serves as a baseline case of the more general amplitude-tunable PASS model in Lemma~1.
\end{remark}
\end{lemma}

\section{SYSTEM MODEL AND PROBLEM FORMULATION}
Based on the physics-based signal model of amplitude-tunable pinching antennas, we extend this framework to multi-waveguide PASS architectures that enable multi-user downlink transmission through waveguide multiplexing \cite{wave_multiplex_7}.
\subsection{Signal Processing for Multiple Waveguides}

We consider a multi-user downlink MISO system with \(K\) users, where \(s_k\) denotes the transmitted symbol for user \(k \in \mathcal{K} \triangleq \{1,\ldots,K\}\). The PASS system consists of \(N_\text{w}\) parallel dielectric waveguides placed over the user deployment region with the set of waveguides, each extending over tens of meters to facilitate extended LoS communication. Each waveguide contains \(N_\text{a}\) amplitude-tunable pinching antennas yielding a total number of \(N \triangleq N_\text{w}N_\text{a}\) radiating elements. 
The PASS structure is connected to the digital baseband precoder through \(N_\text{w}\) RF chains, corresponding to different waveguides. The baseband processor multiplexes the data streams and generates digital beamforming vectors \(\mathbf{w}_k \in \mathbb{C}^{N_\text{w} \times 1}\) corresponding to streaming of \(s_k\), which are fed into the corresponding waveguides to control the propagation of the guided reference waves. The signal radiated from this setup can be written as
\begin{equation}
\mathbf{x}_k =  \mathbf{G}\mathbf{A}\mathbf{w}_k s_k ,
\label{eq:Eq23}
\end{equation}
where $\mathbb{E}[|s_k|^2]=1$,  $\mathbf{A}$ represents the block diagonal matrix of cascaded pinching-antenna coefficients corresponding to the complex amplitudes of elements. Also, $\mathbf{G}$ is a diagonal matrix capturing the in-waveguide propagation response. Specifically, the elements of $\mathbf{G} \in \mathbb{C}^{N \times N}$ are defined as
\begin{equation}
\mathbf{G}_{(i-1)N_\text{a}+l,\,(i-1)N_\text{a}+l} = g_{i,l},
\label{eq:Eq24}
\end{equation}
where $g_{i,l}$ represents the in-waveguide propagation sampled at the $l$-th pinching antenna on the $i$-th waveguide, defined according to \eqref{eq:Eq13}. The matrix $\mathbf{A} \in \mathbb{C}^{N \times N_\text{w}}$ can be written as
\begin{equation}
\mathbf{A} = \mathrm{blkdiag}(\mathbf{a}_1,\mathbf{a}_2,\ldots,\mathbf{a}_{N_\text{w}}),
\label{eq:Eq25}
\end{equation}
where $\mathbf{a}_i \in \mathbb{C}^{N_\text{a}\times1}$ denotes the vector of complex weights of the elements for the $i$-th waveguide, and elements of $\mathbf{A}$ can be derived using \eqref{eq:Eq14} as
\begin{equation}
a_{i,l} =
\left(\prod_{j=1}^{l-1} T_{11}^{(j,i)}\right) T_{21}^{(l,i)}, 
\quad l=1,\ldots,N_\text{a}.
\label{eq:Eq26}
\end{equation}
The matrix $\mathbf{A}$ in \eqref{eq:Eq25} captures the dynamic reconfigurability of the pinching antennas through different schemes of the proposed control mechanism, as described in Lemma~\ref{AT-PASS} and Lemma~\ref{EPA-PASS}.

\begin{remark}
In prior PASS models, the radiation of individual antennas has typically been controlled either through activation mechanisms \cite{wave_multiplex_2} or through real-valued power allocation factors introduced at the signal-processing level \cite{wave_multiplex_5}. In contrast, the proposed formulation establishes a physics-based representation in which each pinching antenna is associated with a complex coefficient $a_{i,l}$ directly linked to the tunable propagation mismatch. As a result, the matrix $\mathbf{A}$ provides simultaneous control over both the amplitude and phase of the radiated signals, enabling flexible radiation shaping and beam steering. This introduces a completely new DoF for PASS-based hybrid beamforming architectures.
\end{remark}

For the channel model, we consider a LoS link between the antennas and the users based on a near-field model \cite{wave_multiplex_12_atten}. The free-space channel coefficient between the $k$-th user located at position $\mathbf{u}_k$ and the $n$-th pinching antenna of the PASS setup located at $\mathbf{p}_{n}$ is given by
\begin{equation}
h_{k,n}
=
\frac{\lambda_0}{4\pi d_{k,n}}
\exp\!\left(
-j \frac{2\pi}{\lambda_0} d_{k,n}
\right),
\label{eq:Eq27}
\end{equation}
where \(d_{k,n}=\|\mathbf{u}_k-\mathbf{p}_{n}\|\) denotes the distance between the $k$-th user and the $n$-th pinching antenna. Note that, $n$ can be defined for the \(l\)-th antenna on the \(i\)-th waveguide with \(n=(i-1)N_\text{a}+l\). Then, the channel vector can be represented with \(\boldsymbol{h}_k \in \mathbb{C}^{N \times 1}\), i.e., \(\boldsymbol{h_k} \triangleq \begin{bmatrix}
h_{k,1},\ldots, h_{k,n}, \ldots, h_{k,N}
\end{bmatrix}^H\). 
Using \eqref{eq:Eq23} and \eqref{eq:Eq27}, the received signal at User \(k\) can be formulated as
\begin{equation}
\label{eq:Eq28}
y_k = \left(\mathbf{h}_k\right)^{H}
\sum_{i=1}^{K} \mathbf{G}\mathbf{A}\mathbf{w}_i s_i
+ n_k .
\end{equation}
where $n_k \sim \mathcal{CN}(0,\sigma_k^2)$ denotes the additive white Gaussian noise (AWGN) at user $k$ with the noise power \(\sigma_k^2\). Based on this, the signal-to-interference-plus-noise ratio (SINR) of $k$-th User can be expressed as
\begin{equation}
\label{eq:Eq29}
\mathrm{SINR}_k =
\frac{\left|\mathbf{h}_k^{H}\mathbf{G}\mathbf{A}\mathbf{w}_k\right|^2}
{\sum_{\substack{i=1 \\ i\neq k}}^{K}
\left|\mathbf{h}_k^{H}\mathbf{G}\mathbf{A}\mathbf{w}_i\right|^2
+
\sigma_k^2 }.
\end{equation}
\begin{remark}
In \eqref{eq:Eq28} and \eqref{eq:Eq29}, $\mathbf{h}_k$ and $\mathbf{G}$ depend on the spatial positions of the pinching antennas. Consequently, adjusting antenna positions provides control over the effective channel seen by the radiating elements. In contrast, the matrix $\mathbf{A}$ depends solely on the propagation-constant mismatches of the elements. Their optimization enables the design of complex radiation weights, thereby introducing an additional DoF for interference mitigation through effective precoder design.
\end{remark}

\subsection{Problem Formulation}
The objective of this work is to maximize the achievable sum rate of a multi-user PASS system by optimizing digital precoders and configurable parameters of the pinching antennas under different PASS architectures. We first consider the amplitude-tunable PASS architecture introduced in Lemma~\ref{AT-PASS}. We then consider commonly studied schemes based on discrete activation and movable pinching antennas with semi-continuous activation, where elements typically radiate equal powers \cite{wave_multiplex_5}. In this work, the equal-power radiation model is realized via tunable $\boldsymbol{\Delta\beta}$, as introduced in Lemma~\ref{EPA-PASS}, providing a practical and electrically controllable mechanism compared to conventional physics-based approaches \cite{wave_multiplex_2,wave_multiplex_3}.

Let $\mathbf{P} \in \mathbb{R}^{N \times 3}$ denote the matrix of spatial positions of the pinching antennas, where the $n$-th row $\mathbf{p}_n$ corresponds to the position of the $n$-th antenna. Let the tunable phase-mismatch vector be denoted by
\(
\boldsymbol{\Delta \beta} =
\left[
\Delta\beta_{1}, \ldots, \Delta\beta_{N}
\right]^{T}.
\)
Then, the optimization problem can be formulated in a generalized form for different PASS schemes as

\begin{equation}
\label{eq:Eq30}
\begin{aligned}
\max_{\Omega} \quad
& \sum_{k=1}^{K}
\log_2\!\left(
1+
\frac{
\left|\mathbf{h}_k^H(\mathbf{P})
\mathbf{G}(\mathbf{P})
\mathbf{A}(\boldsymbol{\Delta \beta})
\mathbf{w}_k\right|^2
}{
\sum_{\substack{i=1\\ i\neq k}}^{K}
\left|\mathbf{w}_i\right|^2
+\sigma_k^2
}
\right)
\\
\text{s.t.}\quad
& \sum_{k=1}^{K}
\left\|
\mathbf{w}_k
\right\|^2
\le P_{\max},  \\\quad &\boldsymbol{\Delta \beta} \in \mathcal{B},
\end{aligned}
\end{equation}
where the first and second constrains in \eqref{eq:Eq30} corresponds to the total transmitted power constraint on the precoder vector, and the feasible set of the phase-mismatches defined by \eqref{eq:Eq20}, respectively. Also, the optimization variable set $\Omega$ depends on the considered PASS architecture and can be defined follows:
\begin{itemize}

\item \textbf{Amplitude-Tunable PASS:}  
$\Omega = \{\mathbf{w}_k,\boldsymbol{\Delta\beta}\}$ with antenna fixed at initial array positions $\mathbf{P}^0$. Based on Lemma~\ref{AT-PASS}, the amplitude matrix $\mathbf{A}$ is constructed using \eqref{eq:Eq26}, with the corresponding transmission parameters $(T_{11},T_{21})$ for the optimized $\boldsymbol{\Delta\beta}$ are obtained using \eqref{eq:Eq10}.

\item \textbf{Discrete-Activation PASS:}  
$\Omega = \{\mathbf{w}_k,\mathbf{b}\}$ with fixed antenna positions $\mathbf{P}^0$, where $\mathbf{b}\in\{0,1\}^{N}$ denotes the antenna activation vector for $N$ elements.  Then, the vector $\mathbf{b}$ inherently determines $\boldsymbol{\Delta\beta}$ via \eqref{eq:21} and \eqref{eq:Eq22}, and this yields the equal-power radiation model of the activated elements based on Lemma~\ref{EPA-PASS}.

\item \textbf{Movable PASS:}  
$\Omega = \{\mathbf{w}_k,\mathbf{P},\mathbf{b}\}$ with $\mathbf{b}$ represents the activation vector controlling the equal-power radiation model as elaborated in Discrete-Activation case. Then, the activated antennas ($b_n=1$) adjust their positions around their nominal locations by selecting from predefined grid points through the optimization of $\mathbf{P}$.
\end{itemize}
The optimization problem in \eqref{eq:Eq30} under the amplitude-tunable PASS scenario extends conventional formulations by enabling direct optimization of complex antenna radiation weights, thereby introducing a new framework for pinching beamforming. The associated challenge is that the physics-based characterization in \eqref{eq:Eq18} couples amplitudes and phases through a sinc-type response, resulting in a highly non-linear and multi-modal objective landscape. Existing PASS frameworks for the same objective primarily consider antenna position control via movability or activation under fixed radiation powers \cite{wave_multiplex_6, wave_multiplex_12_atten, wave_multiplex_1}. These schemes are also encompassed within the proposed framework under different variable sets $\Omega = \{\mathbf{w}_k,\mathbf{b}\}$ or $\{\mathbf{w}_k,\mathbf{P},\mathbf{b}\}$, where radiation power is controlled by tunable phase mismatches, enabling practical implementation via tunable material properties.

\section{PROPOSED SOLUTION}
The solution of problem~\eqref{eq:Eq30} requires the coupled design of the digital precoders $\{\mathbf{w}_k\}_{k=1}^{K}$ and the PASS configuration variables $\boldsymbol{\Delta\beta}$, $\mathbf{b}$, and $\mathbf{P}$. Moreover, the construction of the PASS radiation matrix $\mathbf{A}$ depends on the physics-based coupled-mode formulation through the phase-mismatch parameters $\boldsymbol{\Delta\beta}$, which are restricted to the feasible set $\mathcal{B}$. As a result, problem~\eqref{eq:Eq30} becomes highly non-convex and difficult to solve directly. To address this challenge, we adopt an alternating optimization (AO) approach, in which \eqref{eq:Eq30} is decomposed into two subproblems and solved iteratively. In the first subproblem, for a given PASS configuration, the optimal digital beamforming vectors are obtained using the WMMSE algorithm, enabling effective solution of multi-user beamforming problems in interference-limited environments \cite{wave_multiplex_12_atten,wave_multiplex_6}.  In the second subproblem, the PASS configuration variables are optimized using a genetic algorithm (GA), a population-based search method well suited for handling discrete and continuous design variables \cite{GA_in_Wireless}. 
\subsection{Subproblem 1: Optimizing the Digital Precoder}
\label{Subproblem1_WMMSE}
When the waveguide response matrix $\mathbf{G}$ and the PASS weight matrix $\mathbf{A}$ are fixed, the optimization problem in \eqref{eq:Eq30} reduces to a form analogous to the multi-user MIMO sum-rate maximization problem. In this case, the WMMSE framework can be employed to reformulate the non-convex problem into an equivalent and more tractable form \cite{wmmse}. While global optimality is not guaranteed due to the inherent non-convexity, the WMMSE algorithm is known to converge to a stationary point and has been recently adopted in multi-user PASS systems, demonstrating robust performance \cite{wave_multiplex_12_atten, wave_multiplex_6}. Accordingly, the resulting WMMSE-based optimization problem can be expressed as
\begin{equation}
\label{eq:Eq31}
\begin{aligned}
\max_{\{\mathbf{w}_k,u_k,v_k\}_{k=1}^{K}}
& \quad \sum_{k=1}^{K} \left(\log_2(v_k) - v_k\, e_k(u_k,{\mathbf{w}_k})\right)  \\
\text{s.t.}\quad
& \sum_{k=1}^{K}\|\mathbf{w}_k\|^2 \le P_{\max}, \\
& v_k \ge 0, \quad \forall k \in \mathcal{K},
\end{aligned}
\end{equation}
where $u_k$ and $v_k$ denote auxiliary variables introduced by the WMMSE reformulation, and the mean-square error (MSE) for User k is given by
\begin{equation}
\label{eq:Eq32}
\begin{aligned}
e_k(u_k,\mathbf{w_k})
&=
\left|1-u_k\mathbf{c}_k\mathbf{w}_k\right|^2+ \sum_{j\neq k}
\left|u_k\mathbf{c}_k\mathbf{w}_j\right|^2
+ \sigma^2 |u_k|^2,
\end{aligned}
\end{equation}
where $\mathbf{c}_k=\mathbf{h}_k^{H}\mathbf{G}\mathbf{A}$.
Variables of \eqref{eq:Eq31} can be decoupled into separate blocks, for which the resulting subproblems become concave. Therefore, the block coordinate descent (BCD) method can be employed to solve the problem iteratively \cite{wmmse}. The corresponding update rules of the BCD variables at iteration $t$ are given as
\begin{equation}
\begin{aligned}
u_k^{(t)} 
&= 
\frac{\mathbf{c}_k \mathbf{w}_k^{(t-1)}}%
{\sum_{j=1}^{K} 
\left|\mathbf{c}_k \mathbf{w}_j^{(t-1)}\right|^{2}
+ \sigma^{2}}, \\[6pt]
v_k^{(t)} 
&= 
\big(e_k(u_k^{(t)},\mathbf{w_k}^{(t-1)})\big)^{-1}, \\[6pt]
\mathbf{w}_k^{(t)} 
&= 
v_k^{(t)} u_k^{(t)}
\left(
\sum_{j=1}^{K} 
v_j^{(t)} |u_j^{(t)}|^{2} 
\mathbf{c}_j^{H} \mathbf{c}_j
+ \lambda_\text{p} \mathbf{I}
\right)^{-1}
\mathbf{c}_k^{H},
\end{aligned}
\label{eq:Eq33}
\end{equation}
where $\lambda_\text{p}$ is the Lagrangian multiplier for
equalizing transmit power constraint, which can be set by applying the bisection method \cite{wmmse}. The digital precoders $\{\mathbf{w}_k\}_{k=1}^{K}$ are obtained by iteratively updating $u_k^{(t)}$, $v_k^{(t)}$, and $\{\mathbf{w}_k^{(t)}\}_{k=1}^{K}$ until convergence.

\subsection{Subproblem 2: PASS Optimization}
For a given set of digital precoding vectors of the PASS-based hybrid beamforming architecture, denoted by $\{\mathbf{w}_k\}_{k=1}^{K}$, the optimization problem in \eqref{eq:Eq30} requires determining the PASS array configuration under different architectural scenarios. Depending on the considered PASS scheme, the optimization variables may include the antenna position matrix $\mathbf{P} \in \mathbb{R}^{N \times 3}$, the binary activation vector $\mathbf{b}\in\{0,1\}^{N}$, or the tunable phase-mismatch vector $\boldsymbol{\Delta \beta} \in \mathcal{B}^{N}$. The resulting problem is highly coupled. In particular, the antenna positions matrix $\mathbf{P}$, affect both the waveguide response matrix $\mathbf{G}$ and the user channel vectors $\{\mathbf{h}_k\}_{k=1}^{K}$. The activation vector $\mathbf{b}$ controls which antennas participate in radiation and consequently influences the feasible phase-mismatch parameters $\boldsymbol{\Delta \beta}$ for equal-power radiation. Moreover, the phase-mismatch values $\boldsymbol{\Delta \beta}$ directly determine the complex radiation coefficients through the construction of $\mathbf{A}$. In addition, in all PASS architectures, the feasible set of $\boldsymbol{\Delta \beta}$ is intrinsically coupled to the radiation amplitudes of the elements through the sinc-type relationship defined in \eqref{eq:Eq18}. This nonlinear mapping makes the resulting problem difficult to handle using conventional mathematical optimization techniques. Consequently, the joint design of these variables results in a highly non-convex optimization problem.

Considering the non-convexity and multi-variable structure of the resulting optimization problem, we adopt a metaheuristic optimization framework for configuring the proposed PASS architecture. Recent works have employed population-based particle swarm optimization \cite{multimode_pass, PSO_PASS} and learning-based methods \cite{wave_multiplex_1}, demonstrating effective performance with reasonable computational complexity. The problem \eqref{eq:Eq30} involves  discrete, and continuous optimization variables and is highly coupled, exhibiting a multi-modal objective landscape due to the underlying phase-mismatch based radiation characteristics. To address this, we employ GAs \cite{PYDEAP}, which provide a population-based search mechanism with inherent diversity preservation through mutation and crossover operations, improving exploration of the search space and reducing the likelihood of premature convergence to local optima \cite{GA_in_Wireless}. In this approach, candidate solutions are represented as chromosomes encoding the optimization variables, and their performance is evaluated via a fitness function. Evolutionary operators such as selection, crossover, and mutation are then applied iteratively to evolve the population across successive generations. In the following subsections, we describe how this framework is applied to the proposed PASS configuration problem.

\subsubsection{Chromosome Representation}
Let the GA population consist of $N_{\text{pop}} $ candidate solutions with unique chromosomes 
$\{\mathbf{\chi}_i\}_{i=1}^{N_{\textbf{pop}}}$. Each chromosome encodes the configurable parameters of the PASS architecture \[
\boldsymbol{\chi}_i =
\{\mathbf{b}_i,\boldsymbol{\Delta\beta}_i,\mathbf{P}_i\},
\quad i=1,\ldots,N_{\text{pop}},
\]
where different PASS architectures correspond to different subsets of these variables as described below:
\begin{itemize}
    \item \textbf{Amplitude-Tunable PASS (AT-PASS):} In AT-PASS, tunable phase-mismatches of elements are stored in the chromosomes via real-value encoding in the feasible set of $\mathcal{B}$ as described in \eqref{eq:Eq20} 
    \[
    \boldsymbol{\chi}_i^{\text{AT}}
    =
    \boldsymbol{\Delta\beta}_i,
    \quad
    \boldsymbol{\Delta\beta}_i \in \mathcal{B}^{N},\]
    \item \textbf{Discrete Activation PASS (DAC-PASS):} In DAC-PASS, the activation values of elements are stored in the chromosomes via binary encoding 
    \[
    \boldsymbol{\chi}_i^{\text{DAC}}
    =
    \mathbf{b}_i,
    \quad
    \mathbf{b}_i \in \{0,1\}^{N}.
    \]
    \item \textbf{Movable\footnote{“Movable” denotes configurable antenna placement along the waveguide within the optimization, not the real-time physical movement. In practice, this is achieved by electronically selecting active radiation points on a predefined $\lambda/2$ grid within a bounded region around each nominal antenna position.} PASS (MOV-PASS):}
    In MOV-PASS, each chromosome encodes antenna activation and displacement along the $z$-axis via binary encoding as
    \[
    \boldsymbol{\chi}_i^{\mathrm{MOV}}=[\boldsymbol{\Delta z}_i,\mathbf{b}_i],
    \]
    where $\mathbf{b}_i\!\in\!\{0,1\}^{N}$ denotes the antenna activation vector and 
    $\boldsymbol{\Delta z}_i\!\in\!\big[-\tfrac{\Delta z^{\mathrm{MAX}}}{2},\,\tfrac{\Delta z^{\mathrm{MAX}}}{2}\big]^{N}$ represents the antenna displacements. 
    The displacements are quantized with resolution $\lambda/2$ using 
    $N_\text{q}=\left\lceil \log_2\!\left(\frac{\Delta z^{\mathrm{MAX}}}{\lambda/2}\right)\right\rceil$ bits. 
    The resulting antenna positions are obtained as $z_n=z_n^{0}+\Delta z_n$ for $n=1,\ldots,N$, which defines the antenna position matrix $\mathbf{P}\in\mathbb{R}^{N\times3}$ with rows $\mathbf{p}_n=[x_n,y_n,z_n]^T$, where $z_n^{0}$ denotes the initial array configuration and $\Delta z^{\mathrm{MAX}}=d-\lambda/2$ ensures that adjacent antennas with element spacing $d$ do not overlap.
    
\end{itemize}

\subsubsection{Fitness Function Evaluation}
Following the problem formulation in \eqref{eq:Eq30}, the fitness value of each chromosome is defined as the achievable sum-rate of the users for fixed digital precoding vectors $\{\mathbf{w}_k\}_{k=1}^{K}$. The fitness function is expressed as
\begin{equation}
F(\boldsymbol{\chi})
=
\sum_{k=1}^{K}
\log_2\!\left(
1+
\frac{
\left|
\mathbf{h}_k^{H}(\boldsymbol{\chi})
\mathbf{G}(\boldsymbol{\chi})
\mathbf{A}(\boldsymbol{\chi})
\mathbf{w}_k
\right|^2
}{
\sum_{\substack{j=1\\ j\neq k}}^{K}
\left|
\mathbf{h}_k^{H}(\boldsymbol{\chi})
\mathbf{G}(\boldsymbol{\chi})
\mathbf{A}(\boldsymbol{\chi})
\mathbf{w}_j
\right|^2
+
\sigma_k^2
}
\right).
\label{eq:Eq34}
\end{equation}

At each generation, the GA evaluates the fitness function in \eqref{eq:Eq34} for all candidate chromosomes 
$\{\boldsymbol{\chi}_i^{\mathrm{OPT}}\}_{i=1}^{N_{\text{pop}}}$, where $\mathrm{OPT}\in\{\mathrm{AT},\mathrm{DAC},\mathrm{MOV}\}$ denotes the considered PASS architecture. 
For each chromosome, the system matrices $\mathbf{A}$, $\mathbf{G}$, and the channel vectors $\mathbf{h}_k$ are constructed according to the encoded configuration. 
In particular, the chromosomes of AT-PASS and DAC-PASS determine the radiation matrix $\mathbf{A}$ through the corresponding phase-mismatch parameters, while MOV-PASS additionally modifies the antenna position matrix $\mathbf{P}$, which affects both $\mathbf{G}$ and the channel vectors $\{\mathbf{h}_k\}_{k=1}^{K}$. 
For DAC-PASS and MOV-PASS, the phase-mismatch vector $\boldsymbol{\Delta\beta}$ required for constructing $\mathbf{A}$ is computed analytically using Lemma~\ref{EPA-PASS}.

\subsubsection{Genetic Operators}
Once the fitness values are computed, the GA applies evolutionary operators, i.e., selection, crossover, and mutation, to the chromosomes $\{\boldsymbol{\chi}_i^{\mathrm{OPT}}\}_{i=1}^{N_{\text{pop}}}$ to generate new populations and explore the solution space. In our implementation, the crossover and mutation rates are set to $0.6$ and $0.3$, respectively, with a population size of $N_{\text{pop}}=100$ and $N_{\text{gen}}$=200 generations. Elitism is also employed to retain the best individual in each generation, ensuring monotonic improvement of the best fitness value.

\subsection{Alternating Optimization Framework and Complexity}
The overall procedure for the solution of \eqref{eq:Eq30} can be summarized within an AO framework, as presented in Algorithm~\ref{alg:AO_PASS}. In each AO iteration, the PASS configuration is updated through the GA stage, which determines the corresponding system matrices and channel responses. Subsequently, the digital precoders are updated once per AO iteration using the WMMSE algorithm for the configuration associated with the best chromosome obtained from the GA population. For fixed system matrices $\mathbf{A}$, $\mathbf{G}$, and channel vectors $\{\mathbf{h}_k\}_{k=1}^{K}$, the WMMSE step computes beamforming vectors that are optimal for the given configuration. In the GA stage, elitism is employed to retain the best individual across generations, which ensures non-decreasing best fitness values over the $N_{\text{gen}}$ generations. Consequently, the proposed AO framework provides iterative improvement of the objective value and is terminated when convergence or a maximum number of iterations is reached. The computational complexity of each AO iteration is dominated by the WMMSE update and the GA-based fitness evaluations, resulting in an overall complexity on the order of $\mathcal{O}(K N_\text{w}^3 + N_{\text{gen}} N_{\text{pop}} K N_\text{w}N_\text{a})$.

\begin{algorithm}[t]
\caption{Proposed AO Framework for Solving \eqref{eq:Eq30}}
\label{alg:AO_PASS}
\footnotesize
\begin{algorithmic}[1]

\State \textbf{Input:} PASS scheme $\mathrm{OPT}\!\in\!\{\mathrm{AT},\mathrm{DAC},\mathrm{MOV}\}$, population size $N_{\text{pop}}$, generations $N_{\text{gen}}$
\State Initialize chromosome population $\{\boldsymbol{\chi}_i^{\mathrm{OPT}}\}_{i=1}^{N_{\text{pop}}}$

\Repeat

\State \textbf{1) Digital Precoder Update}
\State Construct $\mathbf{A},\mathbf{G}, \mathbf{h}_k,$ from best chromosome
\State Run WMMSE (BCD) updates in \eqref{eq:Eq33} to obtain $\{\mathbf{w}_k\}_{k=1}^{K}$

\State \textbf{2) PASS Configuration Update (GA)}
\For{$g=1,\ldots,G$}
\State Decode chromosomes and construct $\mathbf{A},\mathbf{G},\mathbf{h}_k$
\State Evaluate $F(\boldsymbol{\chi}^{\mathrm{OPT}}_i),\; i=1,\ldots,N_{\text{pop}}$ using \eqref{eq:Eq34}
\State Apply selection, crossover, mutation, and elitism
\EndFor

\Until{convergence}

\State \textbf{Output:} Optimized PASS configuration and precoders

\end{algorithmic}
\end{algorithm}

\section{NUMERICAL RESULTS}

\begin{figure}[!t]
    \centering
    \includegraphics[width=0.45\textwidth]{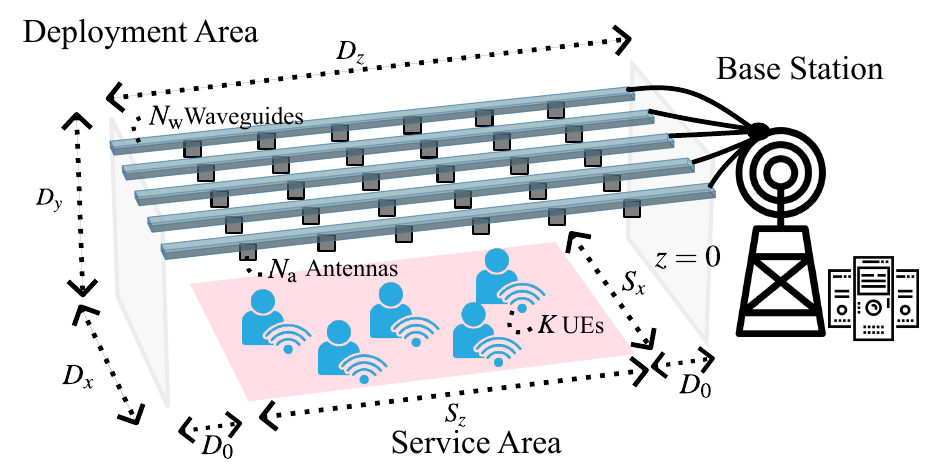}
    \caption{PASS-based hybrid beamforming architecture and the simulation setup.}
    \label{fig:simulation_setup}
\end{figure}
In this section, numerical results are presented to evaluate the performance of the proposed PASS-based hybrid beamforming systems. We consider a downlink multiuser MISO scenario where the hybrid beamforming architecture based on the proposed PASS system is deployed as illustrated in Fig.~\ref{fig:simulation_setup}. In this setup, the digital precoding is performed at the base station (BS), which is located at $z=0$. The PASS structure is mounted on the ceiling of the deployment region of size $D_z \times D_x \times D_y$ at height $y=D_y$ and extends along the $zx$-plane. The system serves $K$ user equipments (UEs) randomly distributed over the service area of size $S_z \times S_x$ on the ground plane ($y=0$) in the $zx$-plane. The waveguides extend along the $z$-direction with length $D_z$ and are distributed across the service region in the $x$-direction, where $D_x=S_x$. The spacing between adjacent waveguides is given by $d_x = S_x/(N_\text{w}-1)$, for $N_\text{w}$ waveguides.  Along each waveguide, $N_\text{a}$ pinching antennas are initially placed at uniformly spaced positions with element spacing $d_z = S_z/(N_\text{a}-1)$ over the service region. Margins of length $D_0$ are reserved at both ends of the waveguide, i.e., $D_z=S_z+2D_0$. In this setup, we evaluate three PASS architectures, namely AT-PASS, DAC-PASS, and MOV-PASS, following the optimization framework described in Algorithm~\ref{alg:AO_PASS}. The tunable wavenumber phase-mismatch values of the pinching antennas are determined according to \eqref{eq:Eq20}. In AT-PASS and DAC-PASS, the pinching antennas remain fixed at their initial positions. In contrast, for MOV-PASS the activated antennas can move on a predefined grid around their initial locations along the $\pm z$ direction with maximum displacement $\Delta z^{\mathrm{MAX}}/2$. This value is set to $\Delta z^{\mathrm{MAX}} = d_z - \lambda_0/2$, ensuring that the spacing between any two consecutive antennas remains at least $\lambda_0/2$. Unless otherwise stated, the simulation parameters of the PASS setup used in the numerical evaluations are summarized in Table~\ref{tab:simulation_parameters}.

As baseline schemes, we first consider a Fixed-PASS architecture, where all pinching antennas along the waveguides radiate with fixed equal powers and no array reconfiguration is performed. As a second baseline, we consider a conventional MISO array with element spacing $\lambda_0/2$ based on fully digital architecture, which enables maximum coherent beamforming gain via the digital precoding. The array follows the same geometry as the PASS setup and is located at the base station ($z=0$), mounted on the ceiling at height $y=D_y$ in the $xz$-plane, with $N_x = N_{\mathrm{w}}$ and $N_z = N_{\mathrm{a}}$ elements, resulting in a total of $N = N_x \times N_z$ antenna elements and RF chains. For both the Fixed-PASS architecture and the conventional MISO array, the digital precoder is obtained using the WMMSE framework. All results are averaged over 500 Monte Carlo realizations of random UE deployments.

\begin{table}[t]
\centering
\caption{Simulation Parameters}
\label{tab:simulation_parameters}
\footnotesize
\begin{tabular}{lc}
\hline
\textbf{Parameter} & \textbf{Value} \\
\hline
Carrier frequency $f$ ($\lambda_0$) & $28\,\mathrm{GHz}$ ($1.07\,\mathrm{cm}$) \\
Noise power $\sigma_k^2$ & $-110\,\mathrm{dBm}$ \\
Number of users $K$ & $5$ \\
Number of waveguides $N_w$ & $5$ \\
Pinching antennas per waveguide $N_\text{a}$ & $6$ \\
Waveguide length $D_z$ & $50\,\mathrm{m}$ \\
Deployment dimensions $(D_x, D_y)$ & $(5\,\mathrm{m},\,10\,\mathrm{m})$ \\
Service area $(S_z \times S_x)$ & $30\,\mathrm{m} \times 5\,\mathrm{m}$ \\
Waveguide edge margin $D_0$ & $10\,\mathrm{m}$ \\
Waveguide spacing $d_x$ & $1.25\,\mathrm{m}$ \\
Antenna spacing $d_z$ & $6\,\mathrm{m}$ \\
Phase-mismatch range $(\Delta\beta L_0)$ & $[0,\; \pi\sqrt{3}]$ \\
Maximum antenna displacement $\Delta z^{\mathrm{MAX}}$ & $d_z-\lambda_0/2 \approx 4.995\,\mathrm{m}$ \\
Waveguide attenuation $\alpha_g$ & $0.08\,\mathrm{dB/m}$ \\
Waveguide refractive index $n_g$ & $1.4$ \\
\hline
\end{tabular}
\end{table}

\subsection{System-Level Performance Evaluation}
\begin{figure}[!t]
    \centering
    \includegraphics[width=0.40\textwidth]{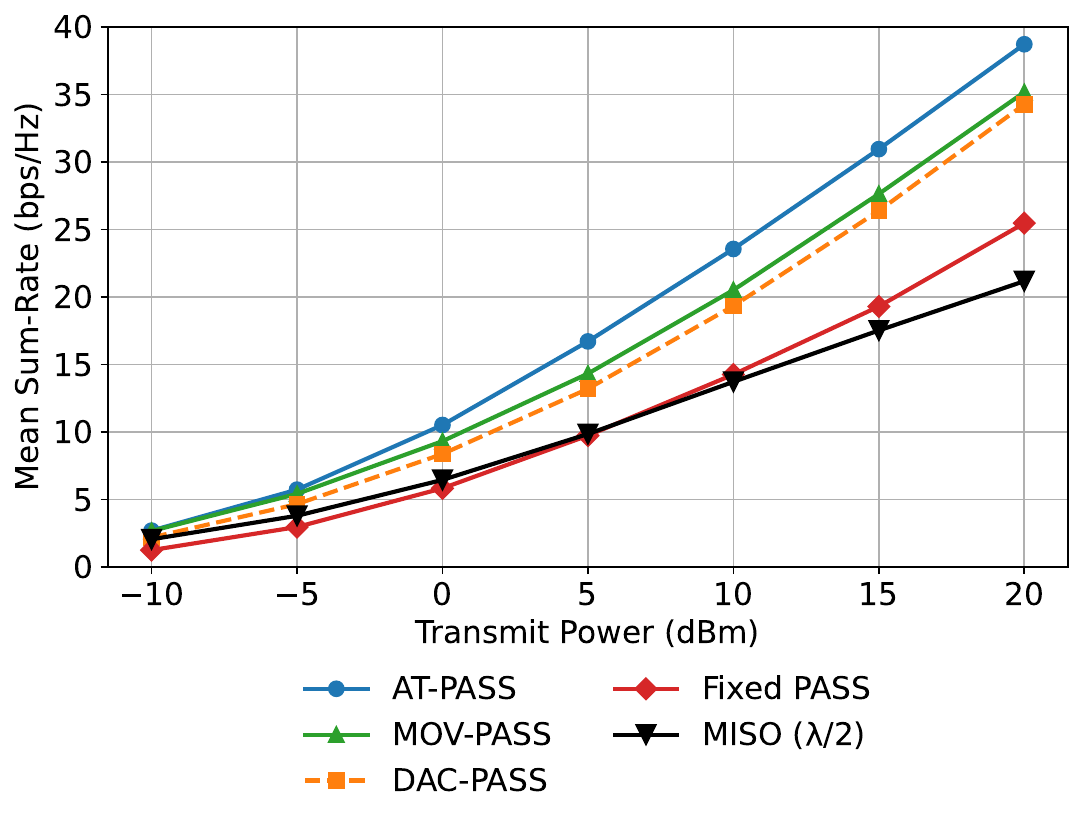}
    \caption{Achievable downlink sum-rate vs transmit power ($P_{\max}$).}
    \label{fig:results_1}
\end{figure}
As a first result, Fig.~\ref{fig:results_1} shows the achievable sum-rate as a function of the transmit power. All three reconfigurable PASS architectures outperform the fixed configurations, including conventional MISO ($\lambda_0/2$) and Fixed-PASS, across all power regimes. This improvement arises from the additional flexibility provided by array reconfiguration, such as antenna activation, movability, and amplitude tunability. The performance gain becomes more pronounced at higher transmit power levels, where the system operates in an interference-limited regime and can better exploit the additional DoF of  offered by reconfigurability. In addition, compared to the conventional MISO ($\lambda_0/2$), PASS systems benefit from their distributed radiating elements. These elements can transmit from locations closer to the users, which reduces path loss and improves the received signal strength. This effect is clearly observed when comparing Fixed-PASS and MISO. Even without any array optimization, Fixed-PASS achieves higher sum-rate at higher transmit power levels ($P_{\max} > 5\,\mathrm{dBm}$). In this regime, the system transitions from noise-limited to interference-limited operation, where the benefits of distributed radiation and improved spatial separability become more pronounced.

As also depicted in Fig.~\ref{fig:results_1}, reconfigurability in PASS systems provides substantial improvements in achievable sum-rate, with performance varying across different architectures.  AT-PASS achieves the highest sum-rate across all power regimes, with the performance gap widening as the system transitions from the low-SNR to the high-SNR region. Specifically, in the low-SNR regime ($P_{\max} = -10\mathrm{dBm}$), where the system is noise-limited, AT-PASS and MOV-PASS achieve comparable performance (approximately 2.7 bps/Hz). In this regime, performance is primarily dictated by received signal power, and MOV-PASS benefits from its ability to reduce path loss via antenna repositioning. In contrast, in the high-SNR regime ($P_{\max} = 20\mathrm{dBm}$), where the system becomes interference-limited, AT-PASS outperforms MOV-PASS (38.7 bps/Hz vs. 35.1 bps/Hz). This is because AT-PASS enables fine-tuned amplitude control of antenna weights, allowing for more effective precoding design to mitigate inter-user interference. Such interference management is limited in MOV-PASS due to its equal-power radiation model, which restricts its ability to fully shape the transmit signals. On the other hand, MOV-PASS consistently outperforms DAC-PASS across all power regimes due to its additional spatial flexibility through position reconfiguration, compared to the fixed-position activation in DAC-PASS.
Having examined the effect of transmit power, we fix the transmit power to $P_{\max} = 15\,\mathrm{dBm}$ for the following results.

\begin{figure}[!t]
    \centering
    \includegraphics[width=0.40\textwidth]{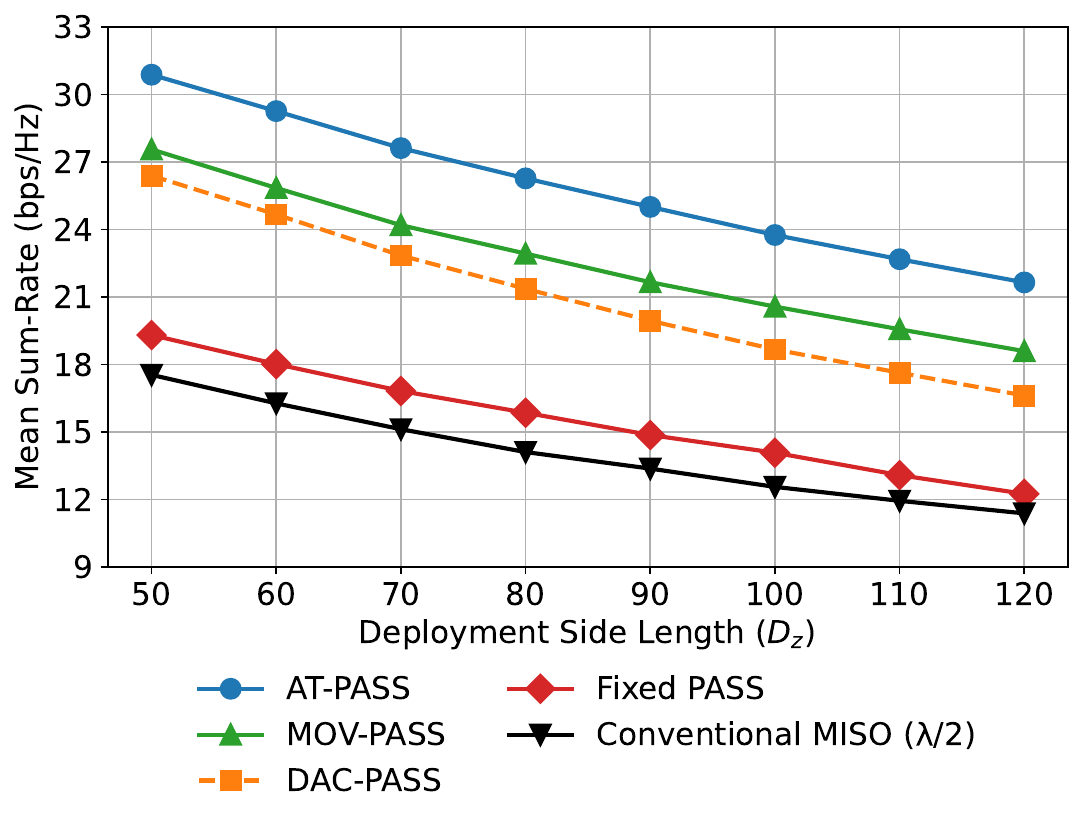}
    \caption{Achievable downlink sum-rate vs deployment area side-length ($D_z$).}
    \label{fig:results_2}
\end{figure}
Next, Fig.~\ref{fig:results_2} shows the achievable sum-rate versus the deployment side length $D_z$, where the waveguide length and the service area length $S_z$ scale according to $D_z = S_z + 2D_0$, with a fixed margin $D_0 = 10\,\mathrm{m}$. While the conventional MISO array benefits from multiplexing gain with $\lambda_0/2$ spacing, PASS architectures exploit distributed antennas to reduce effective propagation distance and enhance spatial diversity.  For $D_z < 100\,\mathrm{m}$, PASS clearly outperforms MISO, with Fixed-PASS achieving around $2$\,bps/Hz gain. As $D_z$ increases, waveguide attenuation reduces this gap to about $1$\,bps/Hz, reflecting a trade-off between MISO multiplexing gain and PASS distributed transmission. 
On the other hand, reconfigurable PASS architectures achieve significantly higher sum-rates than fixed apertures, while both AT-PASS and MOV-PASS consistently outperform DAC-PASS. As the deployment region increases, user channels become less correlated; however, the PASS structure also requires higher reconfigurability due to the increased spacing between elements. In such scenarios, both AT-PASS and MOV-PASS provide improved performance compared to DAC-PASS. Specifically, the performance gap between MOV-PASS and DAC-PASS increases from approximately $1.17$ to $1.98$\,bps/Hz, while for AT-PASS it increases from $4.49$ to $5.03$\,bps/Hz.
The advantage of MOV-PASS over DAC-PASS stems from its ability to reposition antenna elements according to user locations, enabling stronger channel links and improved received signal strength. On the other thand, the performance of MOV-PASS relative to AT-PASS is also limited by the equal-power radiation assumption. AT-PASS enables adjustment of complex weights, providing both amplitude and phase control of the effective precoder. This allows stronger links to receive higher power while adapting other elements accordingly, with phase control providing an additional DoF for link shaping. This explains the performance advantage of AT-PASS over MOV-PASS.

\begin{figure}[!t]
    \centering
    \includegraphics[width=0.40\textwidth]{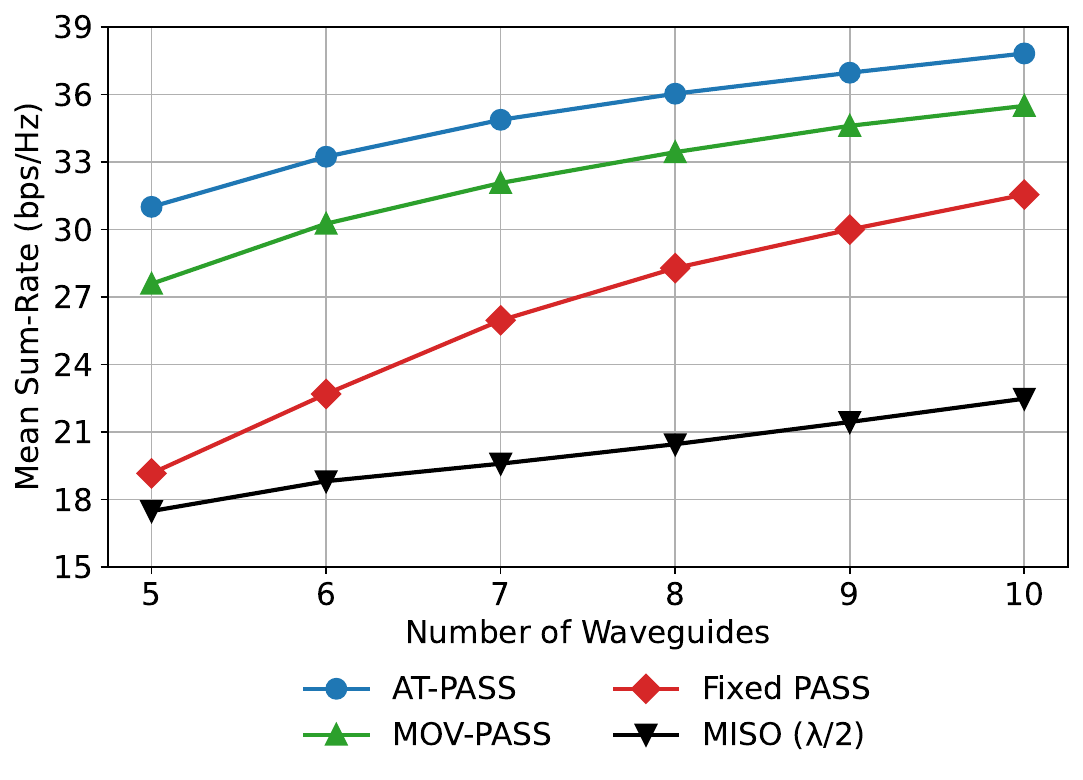}
    \caption{Achievable downlink sum-rate vs waveguide number ($N_\text{w}$).}
    \label{fig:results_3}
\end{figure}
Having established the advantages of AT-PASS and MOV-PASS over DAC-PASS, we focus on AT-PASS and MOV-PASS for the remaining simulations. In Fig.~\ref{fig:results_3}, the effect of $N_\text{w}$ is presented for a $K=5$ simulation setup, where, for the conventional MISO case, this corresponds to an increase in the number of antenna elements in the $x$-direction with $N_x = N_\text{w}$. Hence, the overall dimension of the effective precoder scales similarly for both the conventional MISO array and PASS architectures. As shown in Fig.~\ref{fig:results_3}, this scaling leads to an increase in sum-rate for all schemes. On the other hand, for PASS architectures, increasing $N_\text{w}$ also improves spatial resolution over the region. In addition, the distributed placement of pinching antennas reduces the effective propagation distance to the users. This leads to better scalability of the achievable sum-rate with increasing $N_\text{w}$ compared to the conventional MISO case. As a result, the performance gap increases as $N_\text{w}$ grows. Among PASS architectures, AT-PASS achieves the highest sum-rate across all scenarios. On the other hand, as the contribution of the digital precoder size to the effective precoder becomes more prominent, the limitations of PASS structures with equal-power models become less severe, particularly for the Fixed-PASS architecture. Meanwhile, MOV-PASS achieves more stable performance than Fixed-PASS in all cases due to its additional DoF.

\begin{figure}[!t]
    \centering
    \includegraphics[width=0.40\textwidth]{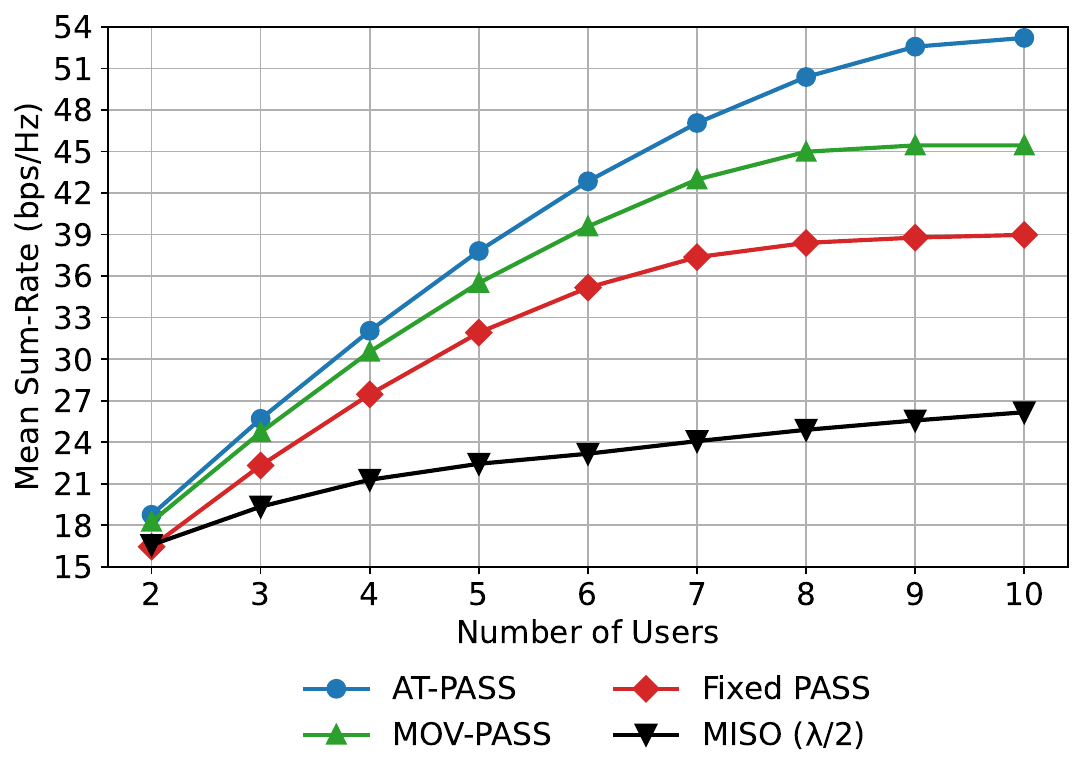}
    \caption{Achievable downlink sum-rate vs Number of Users ($K$) for $N_\text{w}=10$.}
    
    \label{fig:results_4}
\end{figure}
\begin{figure}[!t]
    \centering
    \includegraphics[width=0.40\textwidth]{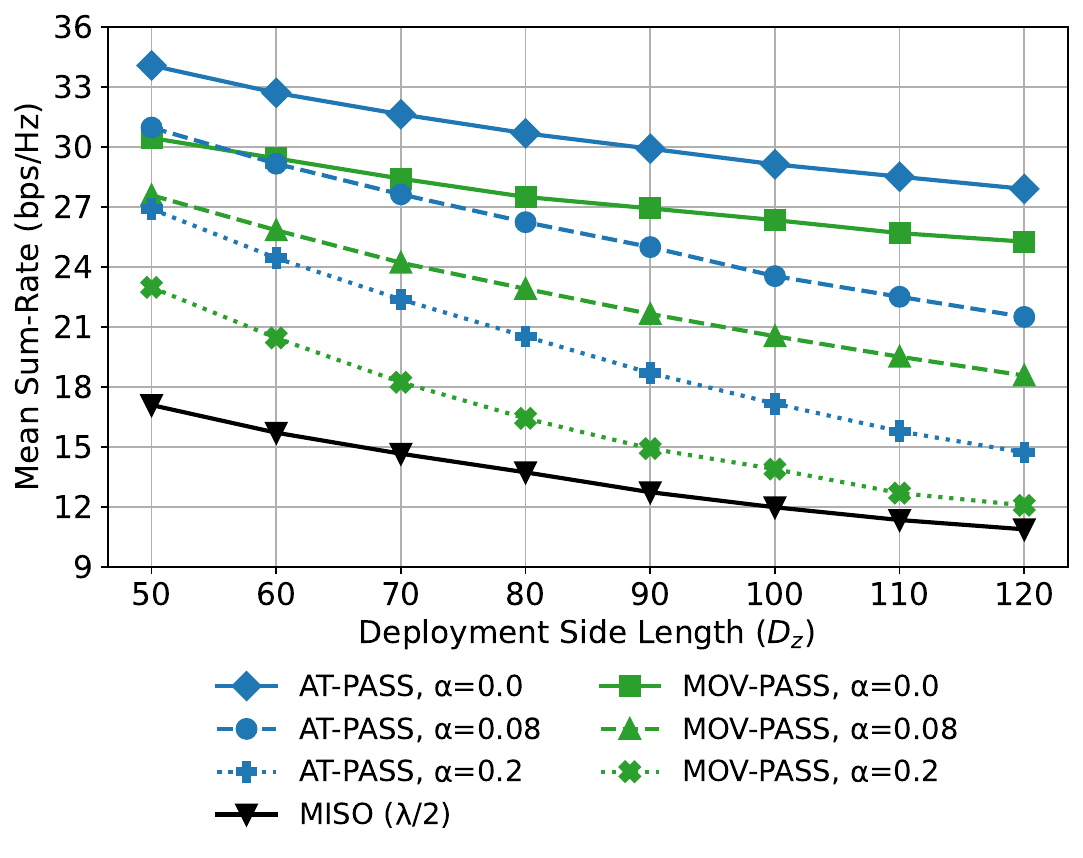}
    \caption{Effect of attenuation for PASS-schemes.}
    \label{fig:results_7}
\end{figure}
Next, Fig.~\ref{fig:results_4} shows the achievable sum-rate versus the number of users $K$ for the setup with $N_x = N_\text{w} = 10$.  Consistent with previous results, all PASS architectures outperform conventional MISO ($\lambda_0/2$), mainly due to reduced path loss. 
For the PASS architectures, the overall sum-rate trend with respect to the number of users follows the well-known behavior. For a low number of users ($K<6$), the sum-rate increases significantly for all PASS architectures, as more DoF become available through the digital precoder. However, beyond this point, the sum-rate begins to saturate due to increased inter-user interference and limited DoF. Furthermore, the results in Fig.~\ref{fig:results_4} again highlight the advantage of amplitude tunability in AT-PASS over MOV-PASS in high-interference scenarios. For $K=2$, AT-PASS and MOV-PASS achieve comparable rates of 18.76 and 18.25\,bps/Hz, respectively, while the gap widens with $K$, reaching 53.2 and 45.4\,bps/Hz at $K=10$, respectively. This gap can be attributed to operation mechanisms of PASS architectures. In low-interference regimes, MOV-PASS can effectively increase the received signal strength in relative to the channel noise by activating antennas closer to the users, where the available power allocated to these active elements.
In contrast, in high-interference regimes, interference management becomes critical. The equal-power radiation in MOV-PASS limits the flexibility of the effective precoder, reducing its ability to suppress inter-user interference. On the other hand, AT-PASS leverages tunable complex weights of individual antennas, enabling greater DoF for the effective precoder. This results in more effective interference mitigation and leads to a noticeable performance gain in achievable sum-rate for high-interference setups.

Finally, Fig.~\ref{fig:results_7} illustrates the effect of the waveguide attenuation $\alpha_\text{g}$ on the achievable sum-rate of AT-PASS and MOV-PASS as a function of the deployment side length. Three cases are considered: a no-loss scenario with $\alpha_\text{g}=0$, a low-loss scenario with $\alpha_\text{g}=0.08$\,dB/m, and a moderate-loss scenario with $\alpha_\text{g}=0.2$\,dB/m \cite{wave_multiplex_12_atten}. As shown in Fig.~\ref{fig:results_7}, increasing $\alpha_\text{g}$ degrades the achievable sum-rate for both PASS architectures. Nevertheless, since free-space propagation losses are substantially higher than typical waveguide losses, PASS-based architectures retain their advantage over conventional MISO arrays. In particular, AT-PASS and MOV-PASS consistently outperform conventional MISO, especially for moderate deployment sizes ($D_z < 100$\,m), where waveguide attenuation is limited and the benefit of reduced effective free-space propagation distance becomes dominant. In this regime, the distributed nature of PASS enables stronger links to users, allowing it to compete with the multiplexing gain of $\lambda_0/2$-spaced MISO arrays. As the deployment size increases, the impact of attenuation becomes more pronounced, reducing the performance gap with MISO. Despite this, AT-PASS maintains the highest performance due to its ability to adapt the amplitude and phase of individual elements, while MOV-PASS provides robust performance through spatial reconfiguration of antennas. These results highlight the critical role of attenuation of PASS systems in multiuser networks, which is often overlooked in prior works.

\subsection{Practical Implementation Considerations}
Having demonstrated the performance advantages of AT-PASS in multiuser downlink networks, we now examine the impact of practical limitations. First, we study the effect of quantization on amplitude tunability, where the complex weights are no longer optimized over a continuous range of phase-mismatches $(\Delta\beta L_0)$ defined in \eqref{eq:Eq20}.  To this end, the optimization algorithm in \ref{alg:AO_PASS} is modified to include binary encoding of phase-mismatches, where the range $(\Delta\beta L_0) \in [0,\; \pi\sqrt{3}]$ is quantized into $2^{N_\text{q}}$ levels, with $N_\text{q}$ denoting the number of bits. Fig.~\ref{fig:results_8} illustrates the sum-rate performance versus transmit power for AT-PASS with $N_q=2,6,10$ in comparison to the ideal continuous case. The results show that for $N_\text{q}=6$ and $10$, the performance closely approaches the ideal continuous case and consistently outperforms MISO across all regimes, demonstrating the practical viability of the proposed system.
\begin{figure}[!t]
    \centering
    \includegraphics[width=0.40\textwidth]{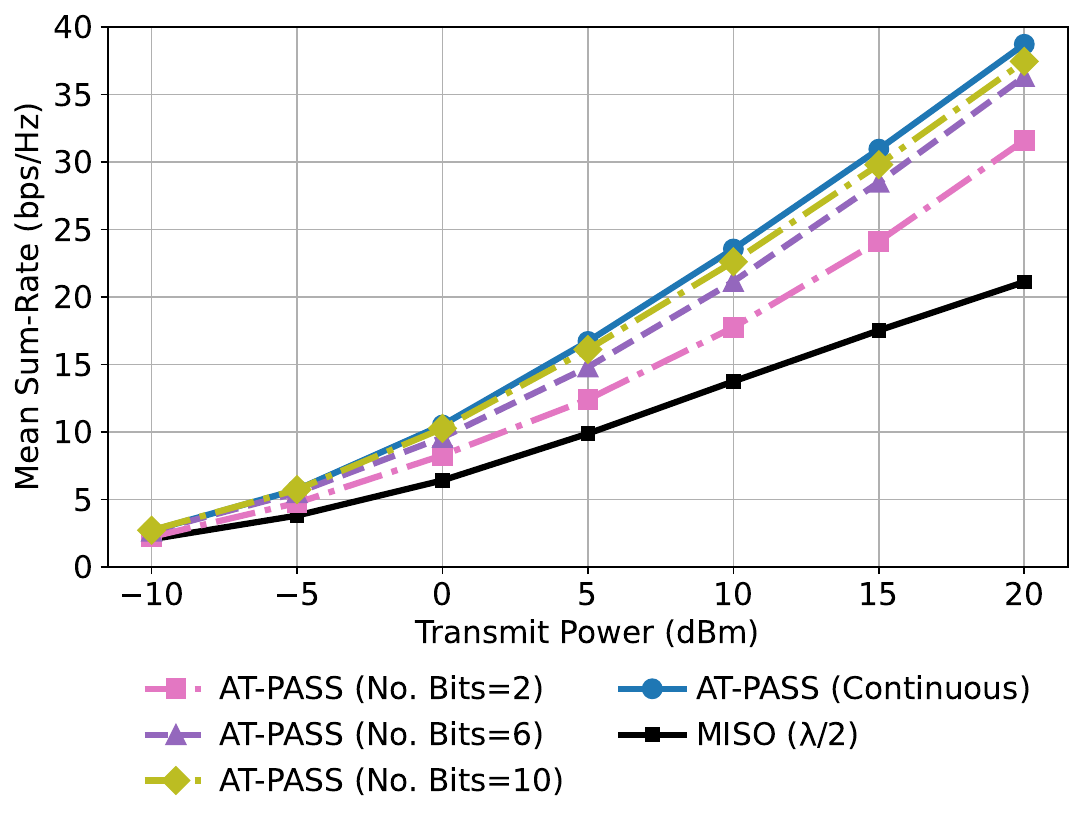}
    \caption{Effect of quantization on amplitude-tunability for AT-PASS.}
    \label{fig:results_8}
\end{figure}
\begin{figure}[!t]
    \centering
    \includegraphics[width=0.40\textwidth]{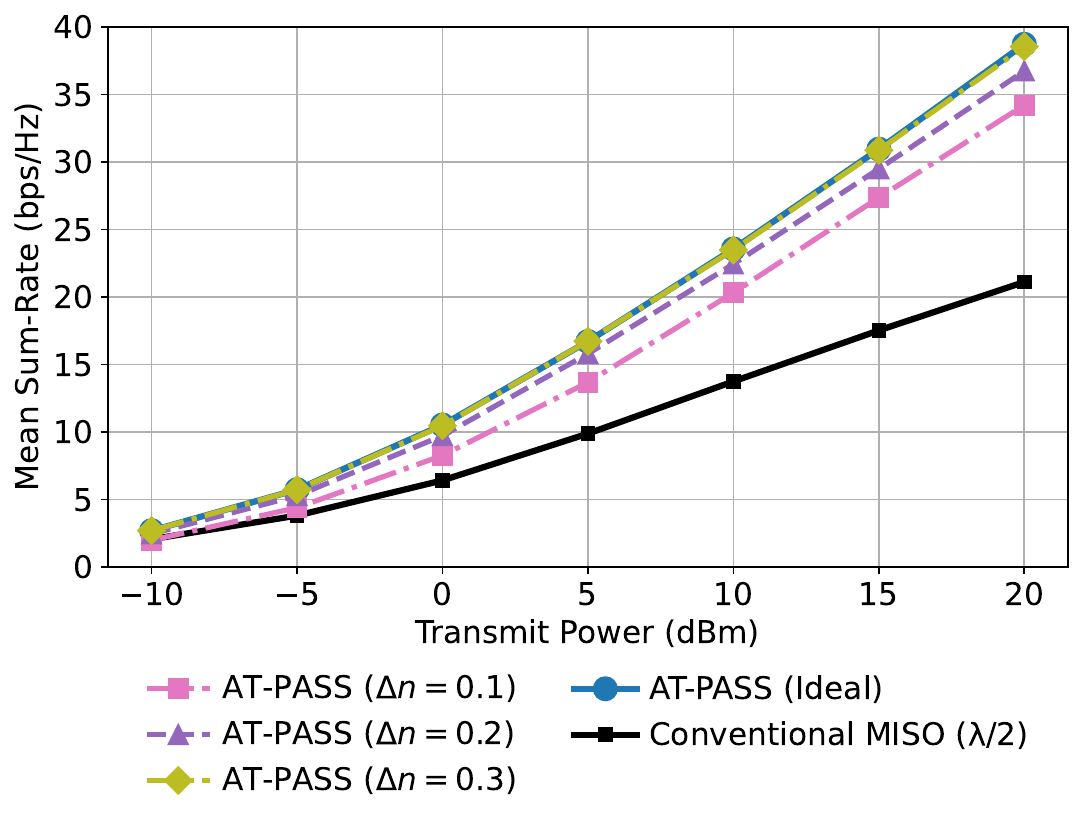}
    \caption{Effect of material tunable range on amplitude-tunability for AT-PASS.}
    \label{fig:results_9}
\end{figure}
\begin{table}[t]
\centering
\caption{Phase shift characteristics for different refractive index variations at pinching length $L_0 = 30\,\mathrm{mm}$.}
\label{tab:phase_shift}
\begin{tabular}{c c c c}
\hline
$\mathbf{\Delta n}$ & $\mathbf{\Delta \beta}$ (rad/m) & $\mathbf{\Delta \phi}$ (rad) & $\mathbf{\Delta \phi}$ (deg) \\
\hline
0.10 & 58.64  & 1.759 & 67.18  \\
0.20 & 117.28 & 3.518 & 134.36 \\
0.30 & 175.92 & 5.278 & 201.54 \\
\hline
\multicolumn{2}{c}{Ideal case} & 5.441 & 311.77 \\
\hline
\end{tabular}
\end{table}
Finally, we investigate the impact of material limitations on the amplitude tunability of the proposed AT-PASS architecture. In the ideal model, the feasible phase-mismatch range is given by $(\Delta\phi=\Delta\beta L_0)\in[0,\;\pi\sqrt{3}]$, enables continuous amplitude control between 0 and 1, as characterized by the power-transfer function in \eqref{eq:Eq12}. In practice, this range is constrained by material properties, particularly the achievable variation in permittivity, leading to a maximum tunable phase shift $\Delta\phi^{\max}$. This limitation can be related to an effective refractive index variation $\Delta n$, where the tunable phase-mismatch of our proposed model can be expressed as $\Delta\beta = k_0 \Delta n$. Practical materials typically can achieve $\Delta n \in [0.1,\,0.3]$ \cite{lqc_2}. To evaluate this effect, we consider $\Delta n = 0.1, 0.2, 0.3$ with $L_0 = 30\,\mathrm{mm}$, and summarize the corresponding phase-tunability ranges in Table~\ref{tab:phase_shift}. The resulting maximum phase shifts $\Delta\phi^{\max}$ are then used for defining the feasible set of tunable phase-mismatches defined in \eqref{eq:Eq20} within Algorithm~\ref{alg:AO_PASS}. Fig.~\ref{fig:results_9} presents the corresponding sum-rate performance compared to the ideal case $\Delta\phi^{\max}=\pi\sqrt{3}$. The results show that for $\Delta n \geq 0.2$, the performance approaches that of the ideal case, confirming the practical feasibility of the proposed AT-PASS architecture.

\section{Conclusion}
This paper proposed, for the first time, an amplitude-tunable PASS architecture that enables independent complex-weight control through phase-mismatch manipulation of guided waves. A physics-based hardware model was developed to link tunable material properties with radiation behavior, providing a unified framework that includes conventional equal-power radiation as a special case. To exploit these capabilities, a sum-rate maximization problem for hybrid precoding in multiuser downlink systems was formulated and solved using an alternating optimization framework combining WMMSE-based digital precoding with genetic algorithm-based PASS configuration. Numerical results demonstrate consistent performance gains over conventional MISO and existing PASS schemes, with more pronounced improvements in interference-limited regimes enabled by complex-weight control. The proposed design was also shown to be robust under practical constraints such as limited material tunability and weight quantization, confirming its feasibility. As the results confirm the advantages of complex-weight control of individual antennas, the proposed approach can be applied to a wide range of PASS use cases with potential extensions to tri-hybrid architectures, amplitude-tunable movable PASS, and uplink scenarios.

\appendix[Proof of Lemma 2]

Let $\mathcal{M}\subseteq\{1,\dots,N_\text{a}\}$ denote the set of the $M$ active
pinching antennas, and define the target equal-radiation power as
$P^{\mathrm{equal}}=1/M$. Let the required power-transfer value of the
$m$-th active antenna be denoted by
$\mathcal{T}(\Delta\beta_m)=\delta_m^2$.
Then, the equal-power radiation constraint can be expressed as
\begin{equation}
P^{\mathrm{equal}}
=
\delta_m^2\, e^{-2\alpha_g z_m}
\prod_{i=1}^{m-1}
\left(1-\delta_i^2\right),
\quad m=1,\ldots,M .
\label{eq_app:Eq1}
\end{equation}

Solving \eqref{eq_app:Eq1} recursively for $\delta_m$, and using the
phase-mismatch power-transfer relation in \eqref{eq:Eq12}, yields
\begin{equation}
\begin{split}
\delta_m^2
&=
\frac{P^{\mathrm{equal}}}{e^{-2\alpha_g z_m}-(m-1)P^{\mathrm{equal}}}
\\
&=
\frac{\pi^2}{4}\,
\mathrm{sinc}^2\!\left(
\frac{1}{2}
\sqrt{
1+\left(\frac{\Delta\beta_m L_0}{\pi}\right)^2
}
\right).
\end{split}
\label{eq_app:delta_transfer}
\end{equation}

Finally, by inverting the right-hand side of \eqref{eq_app:delta_transfer}, the
phase-mismatch required to achieve equal-power radiation from the
active pinching antennas can be obtained as
\begin{equation}
\Delta\beta_m^{\mathrm{equal}}
=
\frac{\pi}{L_0}
\sqrt{
\left(
2\,\mathrm{sinc}^{-1}\!\left(
\frac{2\delta_m}{\pi}
\right)
\right)^2
-1
},
\quad \forall m\in\mathcal{M}.
\label{eq_app:Eq3}
\end{equation}

\bibliographystyle{IEEEtran}
\bibliography{journal_abbreviations,references}

\vfill

\end{document}